\documentclass[aps, pra, amsfonts, twocolumn, a4paper, showpacs,superscriptaddress]{revtex4-1}
\usepackage{graphicx,graphics,epsfig}   
\usepackage{dcolumn}    
\usepackage{bm}         
\usepackage{subfigure}  
\usepackage{times,natbib}
\usepackage{amsmath,amsfonts,amssymb,graphics,graphicx,epsfig,color,times,natbib}
\usepackage{xcolor}
\usepackage{braket}
\usepackage{amsmath}
\usepackage{bbold}

\graphicspath{{graphics/}}
\usepackage{mathtools}
\usepackage{inputenc}
\usepackage[english]{babel}
\usepackage{amsthm}
\newtheorem*{theorem}{Theorem}
\newtheorem*{lemma}{Lemma}

\begin{document}

\title{Efficient detection of useful long-range entanglement in imperfect cluster states}
\author{Thomas Nutz}
\email{nutzat@gmail.com}
\affiliation{Controlled Quantum Dynamics Theory Group, Imperial College London, London SW7 2AZ, United Kingdom}
\author{Antony Milne}
\affiliation{Controlled Quantum Dynamics Theory Group, Imperial College London, London SW7 2AZ, United Kingdom}
\affiliation{Department of Computing, Goldsmiths, University of London, New Cross, London SE14 6NW, United Kingdom}
\author{Pete Shadbolt}
\author{Terry Rudolph}
\affiliation{Controlled Quantum Dynamics Theory Group, Imperial College London, London SW7 2AZ, United Kingdom}
\date{\today}

\begin{abstract}
Photonic cluster states are a crucial resource for optical quantum computing. Recently a quantum dot single photon source has been demonstrated to produce strings of photons in a linear cluster state, but high photon loss rates make it impossible to characterize the entanglement generated by conventional methods. We present a benchmarking method for such sources that can be used to demonstrate useful long-range entanglement with currently available collection/detection efficiencies below $1 \%$. Measurement of the polarization state of single photons in different bases can provide an estimate for the three-qubit correlation function $\braket{ZXZ}$. This value constrains correlations spanning more than three qubits, which in turn provide a lower bound for the localizable entanglement between any two qubits in the large state produced by the source. Finite localizable entanglement can be established by demonstrating $\braket{ZXZ}>\frac{2}{3}$. This result enables photonic experiments demonstrating computationally useful entanglement with currently available technology.
\end{abstract}

\maketitle

\section{Introduction}
Measurement-based quantum computation (MBQC) \cite{raussendorf_measurement-based_2003} has become a promising candidate for the most resource-efficient way to build a universal quantum computer. The greatest challenge of MBQC is the generation of a sufficiently large entangled resource state. This step is critical because only certain types of multi-qubit entanglement are known to enable universal quantum computation \cite{gross_measurement-based_2007}, most prominently cluster state entanglement \cite{briegel_persistent_2001,wei_two-dimensional_2012}.

Photonic systems have been proposed to generate cluster states \cite{nielsen_optical_2004,browne_resource-efficient_2005}. Recently a complete architecture for a linear optical quantum computer has been developed that relies on the probabilistic fusion of many small entangled states into one large cluster state \cite{gimeno-segovia_three-photon_2015}. In this proposal, the generation of the entangled resource state requires only the generation of many maximally entangled three-photon states, at the cost of a significant overhead of single-photon detection measurements and classical information processing. This overhead could be dramatically reduced by developing deterministic ``machine gun'' sources of photonic cluster states \cite{lindner_proposal_2009,economou_optically_2010}.

As experimental work \cite{schwartz_deterministic_2016} makes progress towards realizing the cluster state source proposed in \cite{lindner_proposal_2009}, the challenge of benchmarking such systems comes to the fore. The state produced will certainly be affected by errors and one needs to quantify to what extent the functionality of the state as a resource for MBQC is affected. Since any useful resource state spans many qubits, any benchmarking method with exponential scaling, such as standard quantum state tomography, is out of the question. In fact a benchmarking method for cluster state sources has been presented in \cite{schwarz_efficient_2012} which requires a number of single-photon measurements that scales only linearly in the size of the state being benchmarked. We briefly review this method before stating our result.

Localizable entanglement ($\mathrm{LE}^{i,j}(\rho)$) can be taken as the figure of merit of a cluster state approximation. It is defined with respect to two qubits $i$, $j$ in a state $\rho$ of $n$ qubits as the maximum entanglement between $i$ and $j$ that can be obtained by means of a sequence of single-qubit measurements on the $n-2$ other qubits in $\rho$ \cite{popp_localizable_2005}. The maximum is taken over all measurement bases and we take concurrence to be the entanglement measure. The task of benchmarking a cluster state source is accomplished by lower bounding the localizable entanglement of the state that it produces. 

Particular triplets of expectation values $\braket{B_i},\ i=1,2,3$ provide a lower bound to the localizable entanglement, as shown in Appendix A. We refer to this bound as the \textit{direct bound} in the following. The operators $B_i$ are tensor products of Pauli operators and their expectation values can be measured as correlations between outcomes of single-qubit measurements in the corresponding Pauli bases. \cite{schwarz_efficient_2012} proposes a simple experimental setup to measure these expectation values for photonic systems such as the one proposed in \cite{lindner_proposal_2009}. This benchmarking method, however, is limited by the fact that the time required to measure $\braket{B_i}$ increases exponentially with the number of photons in the state to be benchmarked and the inverse of the collection/detection efficiency. Hence only small states can be benchmarked using the direct bound in currently feasible experiments.

The method presented here overcomes this limitation. We find that if the expectation values $\braket{Z_{i-1}X_iZ_{i+1}}$ on any three neighboring qubits in a linear state $\rho$ are no smaller than some value $\braket{ZXZ}$, then
\begin{equation}
\mathrm{LE}^{j,j+k}(\rho) \geq \max \{ 0, 1-(k+1)(1-\braket{ZXZ}) \}.
\label{eq: intro result}
\end{equation}
Hence measurement of $\braket{ZXZ}$ alone can suffice to establish localizable entanglement across many qubits, as illustrated in Fig. \ref{fig: intro}.

Furthermore, the expectation values $\braket{Z_{i-1}X_iZ_{i+1}}$ give information on the usefulness of a state for practical quantum information processing tasks. For instance we find that a linear state $\rho$ with expectation values $\braket{Z_{i-1}X_iZ_{i+1}} \geq \braket{ZXZ}$ enables a quantum teleportation channel across $n-2$ qubits of fidelity
\begin{equation}
F_T \geq 1 - \frac{n}{3}(1-\braket{ZXZ}).
\label{eq: fid bound}
\end{equation}
In the following we presents proofs of the two results \eqref{eq: intro result} and \eqref{eq: fid bound} and describe their applicability.

\begin{figure}
\includegraphics[width=\linewidth]{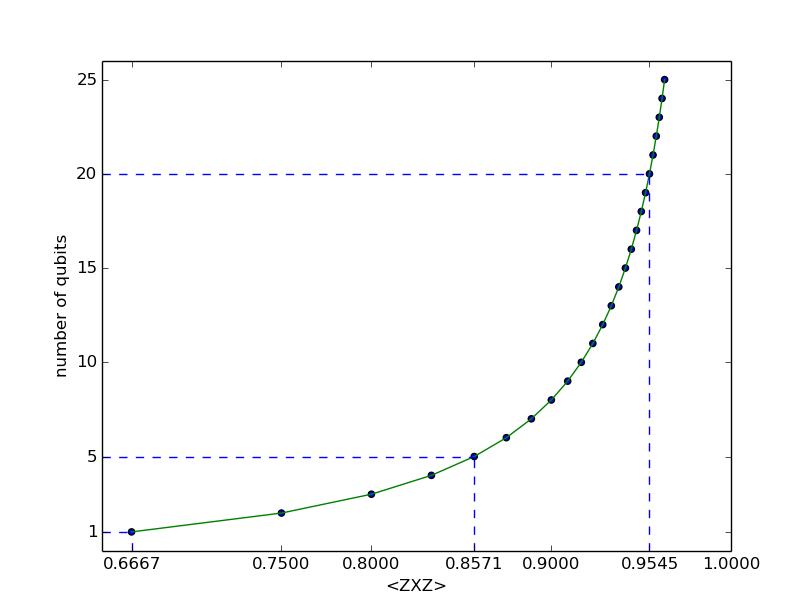}
\caption{A sufficiently high value of $\braket{ZXZ}$ guarantees non-zero LE across many qubits. For instance a value of $\braket{ZXZ}=0.8571$ ($0.9545$) is sufficient to demonstrate non-zero LE across $5$ ($20$) qubits. The curve shown is obtained by setting \eqref{eq: intro result} to zero and plotting $k-1$ (number of measured qubits) vs. $\braket{ZXZ}$.}
\label{fig: intro}
\end{figure}

\section{The cluster backbone and translational invariance}
Cluster states can be elegantly defined in the stabilizer formalism \cite{nielsen_quantum_2010}. The linear cluster state on $n$-qubits $\ket{\Psi_{CS}}$, for instance, is the unique eigenstate with eigenvalue $+1$ of the $n-2$ Pauli operators $K_i=\braket{Z_{i-1}X_iZ_{i+1}}$, where the subscripts label any three adjacent qubits, as well as the two operators $K_1 = \braket{X_1 Z_2}$ and $K_n = \braket{Z_{n-1}X_n}$ at the ends. This collection of $n$ Pauli operators  is called the \textit{stabilizer generator}, and the group of all $2^n$ products of any number of operators in the stabilizer generator is the \textit{stabilizer group} of the linear cluster state. Clearly the cluster state is the $+1$ eigenstate of all operators in the stabilizer group, which we call \textit{stabilizers}. In an imperfect cluster state the expectation values of stabilizers are generally smaller than unity and there may be nonzero expectation values of non-stabilizer operators. We call the collection of expectation values of stabilizers the \textit{cluster state backbone} of a state, because these values determine how similar a state is to the ideal cluster state.

We can furthermore classify stabilizers $K$ according to the number of generators $K_i$ that need to be multiplied to obtain $K$. We denote a stabilizer operator that is the product of $m$ generators as $K^m$. Clearly there are $\binom{n}{m}$ stabilizers $K^m$ where $1\leq m \leq n$.

In the following we make the assumption of \textit{weak translational invariance} (TI) of the state to be characterized. In its strongest form TI requires the reduced density matrix of any number of adjacent qubits to be the same, i.e. independent of where in the $n$ qubit state the segment is found. We only require the weaker assumption that $\braket{K_i}=\braket{K_j}$ $ \forall \ i,j$ and denote $\braket{K_i}=\braket{ZXZ}$ in the following. The assumption of translational invariance is reasonable because the large environments in semiconductor quantum light sources can be assumed to induce mainly Markovian errors, which give rise to TI states. Even the non-Markovian errors modelled in \cite{mccutcheon_error_2014} satisfy TI. 

\section{How does $\braket{ZXZ}$ constrain the cluster backbone?}
While generator expectation values $\braket{ZXZ}=1$ constrain the entire cluster backbone to $\braket{K^m}=1$, the values $\braket{K^m}$ are restricted to finite intervals when $\braket{ZXZ}<1$. To determine a lower bound on these intervals we consider two commuting multi-qubit Pauli operators $P_1$ and $P_2$. The quantity
\begin{equation}\begin{split}
\mathrm{Tr}[\frac{1}{2}(\mathbb{1} - P_1)\frac{1}{2}(\mathbb{1} - P_2)\ \rho\ ] & \\
= \mathrm{Tr}[\frac{1}{4}(\mathbb{1}-P_1-P_2 & +P_1P_2)\ \rho\ ]
\label{Eq: joint prob}
\end{split}\end{equation}
gives the joint probability of measuring both $P_1$ and $P_2$ with outcome $-1$ in state $\rho$ and is therefore non-negative. An alternative argument for its non-negativity is that the product of two commuting positive-semidefinite matrices is positive-semidefinite itself.

Setting the r.h.s. of \eqref{Eq: joint prob} non-negative yields a bound that is fundamental to this work:
\begin{equation}
\braket{P_1P_2}\geq \braket{P_1} + \braket{P_2} -1.
\label{eq: ineq}
\end{equation}

For the case of the cluster backbone for which only the expectation values of the generators $\braket{K_i} = \braket{ZXZ}$ are known, it follows that 
\begin{equation}
\braket{K^m}\geq m(\braket{ZXZ}-1)+1.
\label{eq: n}
\end{equation}

\section{The WC state}

Imagine a device that is capable of emitting a long string of single photons in a TI state $\rho$ with long-range entanglement potentially similar to that of the cluster state. All that can be measured on $\rho$ is the value of $\braket{ZXZ}$. What is the worst-case scenario in terms of long-range entanglement of the state produced, i.e. what is the minimal $\mathrm{LE}^{j,j+k}(\rho)$ consistent with the known $\braket{ZXZ}$?

Note that when we consider the segment of the state $\rho$ spanning qubits $j$ to $j+k$, we mean the state on qubits $j$ to $j+k$ that remains after qubits $j-1$ and $j+k+1$ have been measured in the $Z$ basis with outcomes $+1$ (see Fig. \ref{fig: fig1}). Thus the two expectation values $\braket{Z_{j-1}X_jZ_{j+1}}$ and $\braket{Z_{j+k-1}X_{j+k}Z_{j+k+1}}$ become $\braket{X_jZ_{j+1}}$ and $\braket{Z_{j+k-1}X_{j+k}}$, respectively, which constitute the boundary stabilizers of a cluster state. This construction is necessary because a segment of a perfect cluster state only becomes a cluster state itself after the boundary qubits are ``clipped off'' by $Z$ measurements. For ease of notation, we denote the state of such a $n$-qubit segment by $\rho_n$ in the following and label the qubits from $1$ to $n$.

Eq. \eqref{eq: n} yields a lower bound to any unknown stabilizer expectation value $\braket{K^m}$, and certain triplets $\braket{B_i}$, $i=1,2,3$,  of these expectation values provide a direct bound to $\mathrm{LE}^{1,n}(\rho_n)$. One might now wonder what highest direct bound can be inferred in this way and whether a physical state $\rho_n$ can saturate it. Both these questions are answered by considering a state which we call the \textit{worst-case state} (WC state) on $n$ qubits and define as
\begin{equation}
\rho_n^{WC} = \lambda \ket{C_n}\bra{C_n} + \frac{1-\lambda}{n}\sum_{i=1}^n Z_i \ket{C_n}\bra{C_n} Z_i ,
\label{eq: worst}
\end{equation}
where $\ket{C_n}$ is the linear cluster state on $n$ qubits. The parameter $\lambda$ is chosen such that $\rho_n^{WC}$ is consistent with a given value of $\braket{ZXZ}$, i.e.
\begin{equation}
\lambda = 1 - n \frac{1-\braket{ZXZ}}{2}.
\label{eq: lambda}
\end{equation}

The defining property of the WC state is that it saturates all inequalities given by \eqref{eq: n}:
\begin{equation}\begin{aligned}
\mathrm{Tr}[K^m \rho_n^{WC}] &=  \lambda + \frac{1-\lambda}{n}\sum_{i=1}^n \mathrm{Tr}[K^m Z_i \ket{C_n}\bra{C_n} Z_i] \\
&= \lambda + \frac{1-\lambda}{n}(n-2m) \\
&= m(\braket{ZXZ}-1)+1.
\end{aligned} \label{Eq: saturate}
\end{equation}
The second line is obtained by noting that $K^m$ anticommutes with $m$ of the $Z_i$ operators, yielding a $-1$ for $m$ terms in the sum on the r.h.s. of the first line. The third line follows by substituting $\lambda$ from \eqref{eq: lambda}.

Since $\rho_n^{WC}$ is obviously a physical state and saturates all the inequalities of \eqref{eq: n}, these constraints are in fact a tight bound on the operators in the cluster backbone. Furthermore, $\rho_n^{WC}$ is the state with the lowest possible direct bounds on LE consistent with the given $\braket{ZXZ}$, because any other state cannot have expectation values in the cluster backbone smaller than $\mathrm{Tr}[K^m\rho_n^{WC}]$. Certain triplets of these expectation values, however, guarantee a certain value of LE, as shown in Appendix A. Hence no other state can have LE lower than the direct bound that would be measured in $\rho_n^{WC}$, the value of which is derived in the next section. 

\section{The $\braket{ZXZ}$ bound}

A sequence of $X$ or $Y$ measurements on qubits $2,3,...,n-1$ performed on a linear cluster state results in a maximally entangled state on qubits $1$ and $n$. This ideal resultant two-qubit state can be written as 
\begin{equation}
\tau_{1,n} = \frac{1}{4}(\mathbb{1}\otimes\mathbb{1} + \sum_{i=1}^{3}t_i \ \sigma_i^1 \otimes \sigma_i^n ).
\label{eq: t-state}
\end{equation}
For example, a sequence of $Y$ measurements on an ideal linear cluster state with $+1$ outcomes leads to $\sigma_i^1 = \{Y,Z,X\}$, $\sigma_i^n = \{Z,Y,X\}$, and $\vec{t} = (1,1,1)$ \cite{schwarz_efficient_2012}. We can can always take ${t_i} \geq 0$ by absorbing a minus sign in  $\sigma_i^1$ or $\sigma_i^1$. For states of this form the concurrence is given by \cite{Verstraete2001}
\begin{equation}
C(\tau)=\max\{0,\frac{1}{2}(t_1+t_2 + t_3-1)\}.
\label{eq: ant}
\end{equation}
Given a multi-qubit state $\rho$, the coefficients $t_i$ can be found as the expectation values of the three cluster state stabilizers that commute with every single-qubit measurement (see Fig. \ref{fig: fig1} for an example). In a faulty cluster state $t_i$ will take absolute values smaller than one, and measurement of these three expectation values yields the direct bound (see Appendix A) \cite{schwarz_efficient_2012}. However, as mentioned in the introduction, the operators with expectation values $t_i$ have large support for long strings of photons produced and are therefore hard to measure.

In $\rho_n^{WC}$ every expectation value takes the minimum value allowed by a given $\braket{ZXZ}$, and therefore $C(\tau)$ is minimized for every sequence of measurements and outcomes on qubits $2,3,\cdots,n-1$. In the following we derive the expression for this bound on LE of $\rho_n^{WC}$.

Each $t_i$ is given by an expectation value $\braket{K^{m_i}}$ in $\rho^{WC}_n$ and by \eqref{eq: n} can be written as
\begin{equation}
t_i=m_i(\braket{ZXZ}-1)+1,
\label{eq: t_i}
\end{equation}
where $m_i$ gives the number of stabilizer generators contained in the operator $K^{m_i}$. Using \eqref{eq: ant} we obtain 
\begin{equation}
C(\tau)=\max\{0,\frac{1}{2}(m_1+m_2+m_3)(\braket{ZXZ}-1)\}.
\label{eq: intermediateC}
\end{equation}

\begin{figure}
\includegraphics[width=\linewidth]{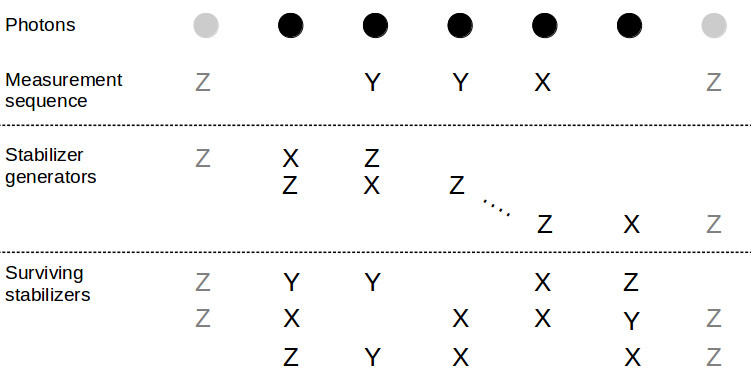}
\caption{Example of a stabilizer description of a measurement sequence that localizes entanglement in a cluster state. The first and last photons (grey dots) are measured out first, such that the state segment has boundary stabilizers $XZ$ and $XZ$. The state $\rho_n$ describes the $n=5$ qubits represented by black dots. Out of the entire stabilizer group only the three operators at the bottom commute with all three measurement projectors and therefore remain in the stabilizer group of the resultant state. It is easy to verify that the surviving stabilizers are a product of three, four, and three generators, respectively (top to bottom), such that $m_i$ take the values ${3,4,3}$.}
\label{fig: fig1}
\end{figure}

The number $m_1+m_2+m_3$ depends only on the number of single-qubit measurements in the $X$ or $Y$ basis and not on the sequence itself. For instance, a single $X$ measurement leads to operators containing one, two, and three $ZXZ$s, while a $Y$ measurement will give three operators each containing two generators. Every further $X$ or $Y$ measurement appends another generator $ZXZ$ to two of the three stabilizers that commute with the measurement sequence, so that
\begin{equation}
m_1+m_2+m_3=4+2m=4+2(n-2),
\end{equation}
where $m$ is the number of $X$ or $Y$ measurements and $n$ the number of qubits in $\rho_n^{WC}$. Hence the LE across $n-2$ qubits in a state $\rho^{WC}_n(\braket{ZXZ})$ is lower bounded by
\begin{equation}
C(\tau_{1,n})=\max \{ 0, 1-n(1-\braket{ZXZ}) \}.
\label{eq: C}
\end{equation}
We refer to this value of concurrence as the $\braket{ZXZ}$ bound in the following. \eqref{eq: intro result} follows by considering a segment $\rho_n$ of $\rho$ as defined above and relabeling qubit $1$, $n$ as qubit $j$, $j+k$, respectively.

\subsection*{Localizable teleportation channel fidelity}
While localizable entanglement certainly is necessary for MBQC, a state of a given LE does not necessarily enable any quantum computation. This discrepancy between entanglement and actual usefulness of a state for quantum information processing tasks has led to the introduction of the fully entangled fraction $F$ as practical measure of quantum information processing significance \cite{grondalski_fully_2002}. The fully entangled fraction (FEF) is defined as the maximum fidelity of a given state $\rho$ with a maximally entangled state $\ket{\phi^+}$ that can be achieved by a local unitary:
\begin{equation}
F(\rho) \coloneqq \max_U \bra{\phi^+} (\mathbb{1} \otimes U) \rho (\mathbb{1} \otimes U^{\dagger}) \ket{\phi^+}.
\label{eq: FEF}
\end{equation}
A two-qubit state $\rho_T$ shared between Alice and Bob enables Alice to teleport a quantum state to Bob with a fidelity of
\begin{equation}
F_T = \frac{1}{3}(1+2F(\rho_T))
\label{eq: FEF teleportation}
\end{equation}
as shown in \cite{Horodecki_1999}. Remarkably, entanglement as measured by concurrence is a necessary but not sufficient condition for a nonzero FEF. There is, for instance, a state of both concurrence and FEF equal to $0.5$. Many copies of this state can be used to distill a number of Bell pairs, yet a single copy offers no advantage over any separable pure state for quantum teleportation ($F_{\ket{\psi}_{sep}}=0.5$) \cite{Horodecki_1999}.
The triplet sum $T \coloneqq t_1 + t_2 + t_3$ not only lower bounds the concurrence of the state $\rho_{t_i}$ with the expectation values $t_i$, but also its FEF:
\begin{equation}
F(\rho_{t_i}) \geq \frac{1}{4}(1+T).
\label{eq: triplet bounds FEF}
\end{equation}

Hence the state $\rho_n^{WC}$ with expectation values $\braket{ZXZ}$ enables a quantum teleportation channel between the first and nth qubit of fidelity
\begin{equation}
F_T \geq 1 - \frac{n}{3}(1-\braket{ZXZ}).
\label{eq: FEF bound}
\end{equation}
Clearly there is no state that is consistent with $\braket{K_i}=\braket{ZXZ}$ of lower triplet sum $T$ than $\rho_n^{WC}$. Hence \eqref{eq: FEF bound} is a tight lower bound on the teleportation fidelity. This statement is stronger than the $\braket{ZXZ}$ bound, because it guarantees that the localizable entanglement present can be used for a single-qubit channel.

\subsection*{Is the $\braket{ZXZ}$ bound tight?}
The $\braket{ZXZ}$ bound rules out the existence of TI states consistent with $\braket{ZXZ}$ of lower LE than the value given in \eqref{eq: C}. The question remains whether this bound is tight, i.e. whether there exists a state which has LE as low as its $\braket{ZXZ}$ bound.   
We believe that the $\braket{ZXX}$ bound is tight because we conjecture that
\begin{equation}
\mathrm{LE}^{1,n}(\rho_n^{WC}(\braket{ZXZ}))=\max \{ 0, 1-n(1-\braket{ZXZ}) \}.
\label{eq: optimality}
\end{equation} 
Appendix B provides strong evidence in support of this conjecture.

\section{Comparison to other characterization methods}
\subsection{The direct bound}
Measurement of $\braket{ZXZ}$ can be accomplished with a simpler experimental setup and fewer photon counts that the direct bound. However the value of LE demonstrated by measuring the direct bound on a source is higher than the LE guaranteed by the $\braket{ZXZ}$ bound for this source, except for the unrealistic scenario where the source produces a WC state. Hence there is a trade-off between saving measurement resources and demonstrating a high value of LE.

An interesting scenario to investigate this trade-off is a cluster state source as proposed in \citep{lindner_proposal_2009} subject to uncorrelated Pauli $Y$-errors on the emitter spin between single-photon emissions. These errors are likely to be the dominant error mechanism in quantum dot implementations of the proposal. Both the direct and the $\braket{ZXZ}$ bound can provide two different figures of merit for such a source. Firstly the number of qubits across which the LE is non-zero can be considered; secondly the value of the LE across a fixed number of qubits constitutes a figure of interest. As shown in Fig. \ref{fig: comparison}, the $\braket{ZXZ}$ bound is considerably below the direct bound for both quantities. The $\braket{ZXZ}$ bound approaches the direct bound as $\braket{ZXZ}\rightarrow 1$ (see inset in Fig. \ref{fig: comparison}).

The $\braket{ZXZ}$ bound is therefore particularly useful for sources with low emission/detection efficiencies and little decoherence of the emitter spin. In that case no direct bound is available but $\braket{ZXZ}$ is close to unity and hence high values of LE across a fixed number of qubits as well as non-zero LE across many qubits can be demonstrated. 

\begin{figure}
\includegraphics[width=\linewidth]{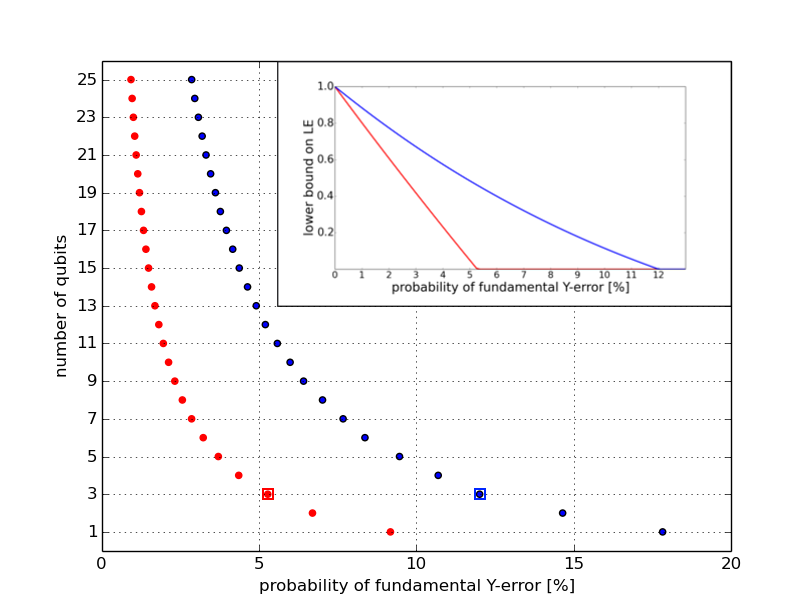}
\caption{The $\braket{ZXZ}$ bound is generally lower than the direct bound. In the main plot the maximum number of qubits across which the direct bound (blue) and the $\braket{ZXZ}$ bound (red) guarantee non-zero LE is shown. The x-axis gives the probability of a Pauli $Y$-error on the emitter spin (fundamental error) before each single-photon emission. The inset shows the value of LE that each bound establishes across three qubits, again plotted against the probability of a fundamental $Y$-error. The highlighted data points in the main plot correspond to the points in the inset plot where each curve hits zero.}
\label{fig: comparison}
\end{figure}

\subsection{Efficient tomography}
The idea of characterizing a cluster state approximation by its stabilizer generator expectation values $\braket{K_i}$ is not new. In fact a bound on the fidelity of a $n$-qubit state $\rho_n$ where $\mathrm{Tr}[K_i\rho_n]=\braket{K_i}$ with the $n$-qubit cluster state $\ket{C_n}$ has been derived in \cite{cramer_efficient_2010}:
\begin{equation}
\braket{C_n|\rho_n|C_n} \geq 1 - \frac{1}{2}\sum_{i=1}^n(1-\braket{K_i}).
\label{eq: fidelity1}
\end{equation}
For the case of $\rho_n$ satisfying weak TI this bound becomes
\begin{equation}
\braket{C_n|\rho_n|C_n} \geq 1 - \frac{n}{2}(1-\braket{ZXZ}).
\label{eq: fidelity}
\end{equation}
This fidelity bound, however, does not directly give information about the long-range entanglement in a state with $\braket{ZXZ}$ approaching one. Clearly two states with the same fidelity with a third, entangled state can themselves have very different entanglement.

Interestingly, the worst-case state $\rho_n^{WC}$ saturates the fidelity bound \eqref{eq: fidelity} and thereby sheds new light onto this result, which was obtained by a seemingly unrelated method. The fidelity of a state $\rho_n$ with $\ket{C_n}$ is given by 
\begin{equation}
\braket{C_n|\rho_n|C_n} = \frac{1}{2^n}\sum_{i=0}^{2^n-1} (-1)^{k_i}\mathrm{Tr}[B_i \rho_n],
\label{eq: Bbone}
\end{equation}
where $B_i$ are the operators in the cluster backbone, $k_i=\frac{1}{2}(1+\mathrm{Tr}[B_i \ket{C_n}\bra{C_n}])$, and $B_0=\mathbb{1}$. Using \eqref{Eq: saturate} for $\mathrm{Tr}[K^m\rho_n^{WC}]$ we find that $\rho_n^{WC}$ saturates the fidelity bound of \eqref{eq: fidelity}:
\begin{equation}\begin{aligned}
\braket{C_n|\rho_n^{WC}|C_n}&=\frac{1}{2^n}\sum_{m=0}^n \binom{n}{m}[m(\braket{ZXZ}-1)+1]\\
&=1 - \frac{n}{2}(1-\braket{ZXZ}),
\end{aligned}
\label{eq: binom}
\end{equation}
where we have used the binomial theorem. 

The WC state is therefore not only the worst-case scenario for LE, but also the worst-case scenario for fidelity given a value of $\braket{ZXZ}$. Moreover since all expectation values in $\rho_n^{WC}$ take on the minimal absolute values allowed by $\braket{ZXZ}$, every state that saturates \eqref{eq: fidelity} must have the same cluster backbone as $\rho_n^{WC}$.

\section{Conclusion}
Our result for the worst-case localizable entanglement in a cluster state approximation substantiates the idea of \cite{cramer_efficient_2010} that an expectation value $\braket{ZXZ}$ provides a meaningful benchmark for an experimentally produced cluster state. We have shown that as $\braket{ZXZ}$ approaches unity, useful entanglement increases at a rate lower-bounded by \eqref{eq: C}.

Measurement of an expectation value $\braket{ZXZ}$ can be accomplished with the simple setup presented in \cite{schwarz_efficient_2012} and emission/detection efficiencies that seem within reach for the experimental implementations of \cite{lindner_proposal_2009}, such as the originally proposed quantum dot system or charged NV-centres \cite{cramer_efficient_2010}. The threshold for establishing a cluster state source capable of producing long-range entanglement is therefore lowered to demonstrating $\braket{ZXZ}=\frac{2}{3}$, which would demonstrate non-zero localizable entanglement across a single qubit.

We gratefully acknowledge fruitful discussions with David Jennings and Ilai Schwarz and the support of the United States Army Research Office and the EPSRC.
\bibliography{Refs}

\appendix
\section{Three expectation values can lower bound the LE}

In \cite{schwarz_efficient_2012} a bound on the LE was derived given three expectation values and assuming zero entanglement fluctuations \cite{popp_localizable_2005} (``outcome-independent entanglement''). In the following we show that this bound holds even in the presence of entanglement fluctuations. We only deal with qubits and projective measurements here.

\begin{theorem}
The expectation values of three multi-qubit Pauli operators $B_{1}$, $B_{2}$, $B_{3}$ lower bound the LE as measured by a convex entanglement measure $E$ between two target qubits $t_1$, $t_2$ in any state $\rho$ of $n$ qubits as
\begin{equation}
\mathrm{LE}^{t_1,t_2}(\rho)\geq E(\rho_B).
\label{Eq: direct bound}
\end{equation}
Here $\rho_B=\frac{1}{4}(\mathbb{1}\otimes \mathbb{1} + \braket{B_1} Z\otimes Y + \braket{B_2} Y \otimes Z + \braket{B_3} X\otimes X)$ and $\{B_i \}$ need to satisfy the following:
\begin{itemize}
\item There exists a collection of $n-2$ single-qubit Pauli operators $\{ P_k \}$, each acting on one of the $n$ qubits except the target qubits $t_1$, $t_2$, such that $[B_l,P_k]=0$ $\forall$ $l,k$. The label $k$ runs from $0$ to $n-2$, while the target qubits are labelled $t_1$, $t_2$.
\item $B_3=B_2B_1$.
\item All three operators $B_i$ have nontrivial support on both qubits $t_1$ and $t_2$, i.e the component of of $B_i$ acting on the $t_1$, $t_2$ subspace is some two-qubit Pauli operator $B_i^t$.
\end{itemize}
\end{theorem}

\begin{proof}
The LE between qubits $t_1$ and $t_2$ in a $n$-qubit state $\rho$ is defined as \cite{popp_localizable_2005}
\begin{equation}
\mathrm{LE}^{t_1,t_2}(\rho)\coloneqq \sup_{\{m\}}\sum_s p_s E(\rho^{t_1,t_2}_{m,s}).
\label{Eq: LE}
\end{equation}
Here $\{m\}$ denotes all possible sequences of local measurements on the $n-2$ qubits other than $t_1$, $t_2$, binary string $s$ gives the outcomes of such a measurement sequence, and $p_s$ the probability of this outcome. $\rho_{m,s}^{t_1,t_2}$ is the state of qubits $t_1$, $t_2$ after a particular measurement sequence $m$ with outcomes $s$.

The state of qubits $t_1$, $t_2$ obtained from $\rho$ by measurement sequence $\{ P_k \}$ with outcomes $s$ is written $\rho^{t_1,t_2}_{\{ P_k \},s}$. The average entanglement obtained by $\{ P_k \}$ is
\begin{equation}
\sum_s p_sE(\rho^{t_1,t_2}_{\{ P_k \},s})
\label{Eq: bound state}
\end{equation}
and clearly lower bounds $\mathrm{LE}^{t_1,t_2}(\rho)$. Furthermore
\begin{equation}
E(\rho^{i,j}_{\{ P_k \},s}) \geq E(\rho_s),
\label{Eq: t-state}
\end{equation}
where $\rho_s\coloneqq \frac{1}{4}(\mathbb{1} + t_1 B_1^t + t_2 B_2^t + t_3 B_3^t)$ (additional terms cannot decrease entanglement \cite{schwarz_efficient_2012}). The coefficients $t_i$ are the expectation values of the corresponding two-qubit Pauli operators $B_i$ in the state $\rho^{t_1,t_2}_{\{ P_k \},s}$.

We now associate strings of $n-2$ bits with operators $Q_q$ on the $n-2$ qubits other than $t_1$, $t_2$ where a $1$ $(0)$ at position $k$ in $q$ means that the component of $Q_q$ acting on qubit $k$ is $P_k$ $(\mathbb{1}_k)$. For instance we have $Q_{01101} = \mathbb{1}_1 \otimes P_2 \otimes P_3 \otimes \mathbb{1}_4\otimes P_5$.

With this notation the coefficients $t_i$ in $\rho_s$ are related to expectation values in the $n$ qubit state $\rho$ (write $\braket{\hat{O}}=\mathrm{Tr}[\hat{O}\rho]$) as
\begin{equation}
t_i= \sum_{q}\frac{(-1)^{s\cdot q}}{2^{n-2}p_s}\braket{B_i^t \otimes Q_q}.
\label{Eq: show}
\end{equation}
The sum runs over all $2^{n-2}$ bit strings $q$ of length $n-2$. $s\cdot q$ denotes the modular sum of all those measurement outcomes $s_i$ where $q_i=1$. One of the terms in the sum of \eqref{Eq: show} corresponds to the known expectation value $\braket{B_i}=\braket{B_i^t \otimes Q_{q_{B_i}}}$.

We can now exploit the convexity of our entanglement measure $E$ and invariance of entanglement under local unitary operations $U_s$ to write 
\begin{equation}\begin{aligned}
\mathrm{LE}^{t_1,t_2}(\rho)\geq \sum_{s}p_s E(\rho_s) &= \sum_s p_s E(U_s \rho_s U_s^\dagger )\\
& \geq E(\sum_s p_s U_s \rho_s U_s^\dagger ).
\end{aligned} \label{Eq: jetza}
\end{equation}
The local unitary $U_s$ can be chosen as a tensor product of two Pauli operators such that it anticommutes with two of the three $B_i^t$. For example if $B_1^t=Z\otimes Y$, $B_2^t=Y\otimes Z$, $B_3^t=X\otimes X$ we could choose $U_s=Z\otimes \mathbb{1}$ for a particular $s$ and have $U_s B_1 U_s^{\dagger} = B_1$, $U_s B_2 U_s^{\dagger} = -B_2$, $U_s B_3 U_s^{\dagger} = - B_3$. We then associate with every choice of $U_s$ a triplet of binary numbers $f_{1,2,3}(s)$ such that $f_i(s)=1$ if $\{ U_s,B_i^t \}=0$ and $f_i(s)=0$ if $[U_s,B_i^t]=0$. Note that $f_1(s)$ and $f_2(s)$ can be chosen freely, while $f_3(s)=f_1(s)\oplus f_2(s)$. The particular choice of $\{ U_s \}$ and associated $f_i(s)$ made here is 
\begin{equation}
f_i(s)=q_{B_i}\cdot s,
\label{Eq: trick}
\end{equation}
where $q_{B_i}$ is the binary string giving the sequence of $\mathbb{1}$ and $P_k$ operators in $B_i$.

Finally the expectation value of $B_i^t$ for the state \mbox{$\rho_{mix}\coloneqq \sum_s p_s U_s \rho_s U_s^\dagger$} can be written using \eqref{Eq: show} as
\begin{equation}
\mathrm{Tr}[B_i^t \rho_{mix}]=\sum_s p_s \sum_q \frac{(-1)^{q\cdot s \oplus f_i(s)}}{2^{n-2}p_s}\braket{B_i^t\otimes Q_q}.
\label{Eq: ExpVal}
\end{equation}
Substituting \eqref{Eq: trick} for $f_i(s)$ we find
\begin{equation}
\sum_{s}(-1)^{(q\oplus q_{B_i})\cdot s} = \begin{dcases*}
	2^{n-2} & when $q = q_{B_i}$, \\
	0 & otherwise.
\end{dcases*}
\label{Eq: binary sum1}
\end{equation}
Hence the functions $f(s,1)$ and $f(s,2)$ can be chosen such that $\rho_{mix}=\frac{1}{4}(\mathbb{1}\otimes \mathbb{1} + \braket{B_1} B_1 + \braket{B_2} B_2 + \braket{B_3} B_3)$. Without loss of generality any triplet of two-qubit Pauli operators satisfying the properties required in the theorem can be chosen as $B_1 = Z\otimes Y$, $B_2 = Y\otimes Z$ and $B_3 = X\otimes X$ such that
\begin{equation}
\mathrm{LE}^{t_1, t_2}(\rho)\geq E(\rho_B).
\label{Eq: finale}
\end{equation}
\end{proof}

\section{Is the $\braket{ZXZ}$ bound tight?}
We provide three items of evidence in support of our conjecture
\begin{equation}\begin{aligned}\mathrm{LE}^{1,n}(\rho_n^{WC}(\lambda))&=\max\{0,1-n(1-\braket{ZXZ})\} \\
&=\max\{0,2\lambda -1\}.\end{aligned}\end{equation} Firstly we show that all equatorial measurement sequences yield entanglement no higher than \eqref{eq: C}; secondly we show analytically that the claim holds for $\rho_4^{WC}$; and finally we present numerics for $\rho_7^{WC}$.

\subsection{Equatorial measurement sequences}

\begin{lemma}Given $\rho_n^{WC}(\lambda)$, no sequence of equatorial measurements can produce a two-qubit state $\rho_{1,n}$ with \mbox{$E(\rho_{1,n})>2\lambda - 1$}.\end{lemma}

\begin{proof}
We first derive a form for the two-qubit state $\ket{\phi_{1,n}}$ resulting from a sequence of equatorial measurements $\mathcal{M}_e$ on qubits $2$ to $n-1$ in the $n$-qubit cluster state $\ket{C_n}$. An equatorial measurement projects onto either of the states $\ket{\varphi,s=0,1}=\frac{1}{\sqrt{2}}(\ket{0}+(-1)^se^{i\varphi}\ket{1})$. Using the expression
\begin{equation}\begin{aligned}
(\ket{\varphi,s}\bra{\varphi,s}_A\otimes \mathbb{1}_B) \mathrm{S}_{AB}(\ket{\Psi}_A\otimes \ket{+}_B) \\
=\ket{\varphi,s}_A\otimes X^sH(\varphi)\ket{\Psi}_B,
\end{aligned}
\label{eq: genHad}
\end{equation} 
where $\mathrm{S}_{AB}$ is the controlled phase gate and
\begin{equation}H(\varphi)= \left( \begin{array}{cc}
1 & e^{i\varphi} \\
1 & -e^{i\varphi} \end{array} \right),\end{equation}
we find that 
\begin{equation}
\ket{\phi_{1,n}}=\mathcal{N}\braket{\vec{\varphi},\vec{s}\ |C_n}=(\mathbb{1}_1 \otimes U_n)\ket{\phi^+}_{1,n}.
\label{eq: eqCluster}
\end{equation}
Here $\ket{\phi^+}$ denotes a maximally entangled state, $\mathcal{N}$ the normalization factor,\begin{equation}\ket{\vec{\varphi},\vec{s}}=\ket{\varphi_2,s_2}_2\ket{\varphi_3,s_3}_3 \cdots \ket{\varphi_{n-1},s_{n-1}}_{n-1},\end{equation} and
\begin{equation}
U_n=X^{s_{n-1}}H(\varphi_{n-1}) \cdots X^{s_{3}}H(\varphi_{3})X^{s_{2}}H(\varphi_{2 }).
\label{eq: U}
\end{equation}
Without loss of generality we can take $s_i=0$ when writing $U_n$ in the following.

Furthermore we find that
\begin{equation}\begin{aligned}
\braket{\vec{\varphi},\vec{s}\ |Z_1|C_n}&\propto (\mathbb{1}_1 \otimes U_nX)\ket{\phi^+}_{1,n} \equiv \ket{\phi_{1,n}^{Z_1}}\\
\braket{\vec{\varphi},\vec{s}\ |Z_j|C_n}&\propto (\mathbb{1}_1 \otimes U_n^j)\ket{\phi^+}_{1,n}\equiv \ket{\phi_{1,n}^{Z_j}}\\
\braket{\vec{\varphi},\vec{s}\ |Z_n|C_n}&\propto (\mathbb{1}_1 \otimes (ZU)_n)\ket{\phi^+}_{1,n}\equiv \ket{\phi_{1,n}^{Z_n}},
\end{aligned}\end{equation}
where $U_n^j$ is given by $U_n$ from \eqref{eq: U} with $s_j=1$ and $s_i=0\ \forall \ i\neq j$. The state $\rho_{1,n}$ resulting from $\mathcal{M}_e$ on $\rho_n^{WC}$ is therefore given by
\begin{equation}
\begin{aligned}
\rho_{1,n}=&\;\mathcal{N}\bra{\vec{\varphi},\vec{s}\,}\rho_n^{WC} \ket{\vec{\varphi},\vec{s\,}}\\
=&\;\lambda \ket{\phi_{1,n}}\bra{\phi_{1,n}} + \frac{1-\lambda}{n}\ket{\phi_{1,n}^{Z_1}}\bra{\phi_{1,n}^{Z_1}} \\
 &+\frac{1-\lambda}{n}\sum_{j=2}^{n-1}\ket{\phi_{1,n}^{Z_j}}\bra{\phi_{1,n}^{Z_j}} +\frac{1-\lambda}{n}\ket{\phi_{1,n}^{Z_n}}\bra{\phi_{1,n}^{Z_n}}.
\end{aligned}
\nonumber
\end{equation}
We now observe that for $1\leq i \leq n$
\begin{equation}
C(\lambda \ket{\phi_{1,n}}\bra{\phi_{1,n}} + (1-\lambda)\ket{\phi_{1,n}^{Z_i}}\bra{\phi_{1,n}^{Z_i}})=2\lambda -1,
\end{equation}
as $\ket{\phi_{1,n}}\bra{\phi_{1,n}}$ and $\ket{\phi_{1,n}^{Z_i}}\bra{\phi_{1,n}^{Z_i}}$ are orthogonal, maximally entangled states.
Therefore, by convexity of concurrence,
\begin{equation}
E(\rho_{1,n})\leq 2\lambda -1.
\label{eq: EqBound}
\end{equation}
\end{proof}

\subsection{Analytics for $\rho_4^{WC}$}
We derive an exact expression for the concurrence of the two-qubit state resulting from two measurements in the $X$-$Z$ plane on the middle qubits in the four-qubit WC state. The result substantiates our conjecture that general measurement sequences on $\rho_n^{WC}$ cannot reach higher entanglement than equatorial ones.

The result of a projective measurement on qubit $i$ in the $X$-$Z$ plane may be written as $\ket{\theta_i}=\cos \frac{\theta_i}{2}\ket{0} + \sin \frac{\theta_i}{2}\ket{1}$. We perform measurements on qubits two and three of $\rho_4^W$, with outcomes parametrized by $\theta_2$ and $\theta_3$ respectively. This yields the two qubit-state
\begin{equation*}
\rho_{1,4}=\frac{1}{4}(\mathbb{1}\otimes \mathbb{1} + \vec{r}\cdot \vec{\sigma}\otimes \mathbb{1} + \mathbb{1} \otimes \vec{s}\cdot \vec{\sigma} + \sum_{i,j=1}^3 T_{ij}\sigma_i \otimes \sigma_j),
\end{equation*}
which has components
\begin{align*}
\vec{r}&=\frac{1+\lambda}{2}\left( \begin{array}{c}
\cos\theta_2 \\
0 \\
\sin\theta_2\cos\theta_3\end{array} \right),\;
\vec{s}=\frac{1+\lambda}{2}\left( \begin{array}{c}
\cos\theta_3 \\
0 \\
\sin\theta_3 \cos\theta_2 \end{array} \right),\\[1em]
T&=\left( \begin{array}{ccc}
\lambda \cos\theta_2 \cos\theta_3 & 0 & \lambda \sin\theta_3\\
0 & (2\lambda -1)\sin\theta_2\sin\theta_3 & 0 \\
\lambda \sin\theta_2 & 0 & 0 \end{array} \right).\end{align*}

The entanglement of a state is invariant under local unitary operations, $\rho_{1,4}\rightarrow(U_1\otimes U_4)\rho_{1,4}(U_1^\dagger\otimes U_4^\dagger)$. This corresponds to the transformations $\vec r\rightarrow O_1 \vec r$, $\vec s\rightarrow O_4 \vec s$, $T\rightarrow O_1 T O_4^\mathrm{T}$, where $O_1$ and $O_4$ are orthogonal matrices \cite{Horodecki1996}. By  choosing $O_1$ and $O_4$ that achieve a signed singular value decomposition of $T$, we perform local unitary operations on $\rho_{1,4}$ that correspond to the transformation
\setlength{\arraycolsep}{0pt}
\begin{equation}\rho_{1,4}\rightarrow
\frac{1}{4}\left( \begin{array}{cccc}
N_+(1+\lambda) & 0 & 0 & S(1-3\lambda) \\
0 & 1-\lambda & S(\lambda-1) & 0 \\
0 & S(\lambda-1) & 1-\lambda & 0 \\
 S(1-3\lambda)& 0 & 0 & N_-(1+\lambda) \end{array} \right),
\label{eq: X-form}
\nonumber
\end{equation}
where $S=\sin\theta_2\sin\theta_3$ and $N_\pm=1\pm\sqrt{1-S^2}$.

The above density matrix is manifestly in the form of an $X$-state \cite{Yu2007}. The concurrence of such a state is a simple function of the density matrix elements, which for us gives
\begin{equation}
C(\rho_{1,4})=\max\{0,\frac{1}{2}(3\lambda - 1)S+\frac{1}{2}(\lambda -1)\}.\label{eq: concSinSin}
\end{equation}
The entanglement that is localized by measurements in the $X$-$Z$ plane on $\rho_4^W$ is therefore clearly maximized by the equatorial measurement sequence, which gives $\sin\theta_2=\sin\theta_3=1$ and hence $S=1$.

\subsection{Numerics for $\rho_7^{WC}$}

\begin{figure}
	\includegraphics[scale=0.4]{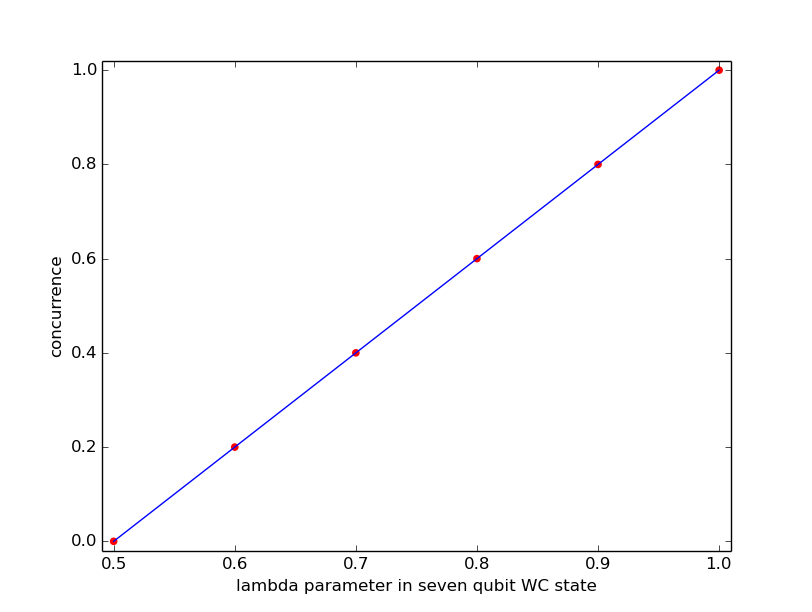}
	\caption{Numerical optimization of the concurrence of the two-qubit state resulting from five single-qubit measurements on $\rho_7^W$. The red dots show the maximum entanglement obtained while the blue line gives the value $2 \lambda -1$ of the $\braket{ZXZ}$ bound.}
	\label{fig: it's linear!}
\end{figure}

We have not found an analytic form for $\mathrm{LE}^{1,n}(\rho_n^{WC})$ when $n>4$. To further investigate whether our conjecture $\mathrm{LE}^{1,n}(\rho_n^{WC})=2\lambda - 1$ holds we therefore perform a numerical optimization.

Using a Nelder-Mead simplex algorithm the following optimization was carried out:
\begin{equation}
\max_{\vec{\theta},\vec{\varphi}}\ C(\mathcal{N} \braket{\vec{\theta}, \vec{\varphi}\ |\ \rho_7^W\ |\ \vec{\theta}, \vec{\varphi}}),
\end{equation}
where $\mathcal{N}$ gives the normalization factor and
\begin{equation}
\ket{\vec{\theta},\vec{\varphi}\,}=\otimes_{i=2}^6\left( \sin\frac{\theta_i}{2}\ket{0}_i+e^{i\varphi_i}\cos\frac{\theta_i}{2}\ket{1}_i\right).
\end{equation}

The results for six different values of $\lambda$ are shown in Fig. \ref{fig: it's linear!}. We find that the optimal measurement angles are $\theta_i=\frac{\pi}{2}$ and that $\varphi_i$ is arbitrary. This yields a value of concurrence $C(\rho_{1,7})=2\lambda -1$, again providing evidence for the conjecture $\mathrm{LE}^{1,n}(\rho_n^{WC})=2\lambda - 1$.

\end{document}